\documentclass[a4paper]{llncs}

\usepackage{amssymb}
\usepackage{amsmath}
\usepackage{stmaryrd}
\usepackage[usenames,dvipsnames]{xcolor}
\usepackage{enumerate}
\usepackage{proof-dashed}
\usepackage{import}
\usepackage{wrapfig}
\usepackage{breakcites}
\usepackage{hyperref}
\usepackage{etex}
\usepackage[all]{xy}
\newcommand{\m}[1]{\ensuremath{\mathsf{#1}}}
\newcommand{\pid}[1]{\m{#1}}

\newcommand{\carr}[1]{\langle #1  \rangle}

\newcommand{\emptyN}{\nil} 
\newcommand{\code}[1]{\texttt{\upshape #1}}
\newcommand{\nil}{\boldsymbol 0}
\newcommand{\dlock}{\boldsymbol 1}
\newcommand{\com}[2]{#1\;\code{-\hspace{-0.3mm}>}\;#2}
\newcommand{\gencom}{\com{\pid p.e}{\pid q}}
\newcommand{\gencomf}{\com{\pid p.e}{\pid q.f}}
\newcommand{\genmulticometa}{(\til \eta)}

\newcommand{\sel}[3]{\com{#1}{#2 [#3]}}
\newcommand{\gensel}{\sel{\pid p}{\pid q}{l}}
\newcommand{\lto}[1]{\mathrel{\stackrel{{\;\;#1\;\;}}{\mbox{\rightarrowfill}}}}
\newcommand{\condlbl}[3]{\eqcom{#1\,}{\,#2}: #3}
\newcommand{\lmto}[1]{\mathrel{\stackrel{#1}{\mbox{\rightarrowfill}}\!\!\raisebox{1ex}{\scriptsize$\ast$}}}
\newcommand{\lleft}{\textsc{l}}
\newcommand{\lright}{\textsc{r}}

\newcommand{\cond}[3]{\m{if}\, #1 \, \m{then} \, #2 \, \m{else} \, #3}
\newcommand{\eqcom}[2]{#1 \! \stackrel{\code{<\!-}}{\code{=}}\! #2}
\newcommand{\gencond}{\cond{\eqcom{\pid p}{\pid q}}{C_1}{C_2}}
\newcommand{\rec}[3]{\m{def} \, #1  =  #2 \, \m{in} \, #3}
\newcommand{\genrec}{\rec{X}{C_2}{C_1}}
\newcommand{\pn}{\m{pn}}
\newcommand{\call}[1]{#1}
\newcommand{\gencall}{\call X}
\newcommand{\pcont}{\mathtt{\ast}}
\newcommand{\til}{\tilde}

\newcommand{\epp}[2]{[\![#1]\!]_{#2}}
\newcommand{\rname}[2]{\ensuremath{\left\lfloor\mbox{{#1}$|${#2}}\right\rceil}}
\newcommand{\bsend}[2]{{#1}!\carr{#2}}
\newcommand{\brecv}[1]{#1?}
\newcommand{\parp}{\, \boldsymbol{|} \, }
\newcommand{\bsel}[2]{{#1}\oplus#2}
\newcommand{\bbranch}[2]{{#1}\&\{{#2}\}}
\newcommand{\precongr}{\preceq}

\newcommand{\smallpar}[1]{\paragraph{#1.}}

\newcommand{\acspar}{\hspace{1.2mm} | \hspace{1.2mm} }
\newcommand{\eval}{\downarrow}

\newcommand{\extract}[1]{(\![{#1}]\!)}

\newcommand{\rwto}{\leadsto}

\newcommand{\bisim}{\sim}

\newcommand{\asend}[2]{{#1}!\carr{#2}}
\newcommand{\arecv}[1]{#1?}
\newcommand{\actor}[3]{#1 \triangleright_{#2} #3}
\newcommand{\asel}[2]{{#1}\oplus#2}
\newcommand{\abranch}[2]{{#1}\&{#2}}
\newcommand{\multicom}[1]{\left(\begin{array}c #1 \end{array}\right)}
\newcommand{\snd}{\m{snd}}
\newcommand{\rcv}{\m{rcv}}

\begin{document}
\pagestyle{plain}

\title{The Paths to Choreography Extraction\thanks{%
Montesi was supported by CRC (Choreographies for Reliable and efficient Communication software), grant 
DFF--4005-00304 from the Danish Council for Independent Research.
Cruz-Filipe and Larsen were supported in part by
the Danish Council for Independent Research, Natural Sciences,
grant DFF-1323-00247.
}}
\author{Lu\'\i s Cruz-Filipe \and
        Kim S. Larsen \and
        Fabrizio Montesi}
\institute{University of Southern Denmark {\tt\{lcf,kslarsen,fmontesi\}@imada.sdu.dk}}

\maketitle

\begin{abstract}
Choreographies are global descriptions of interactions among concurrent components, most notably 
used in the settings of verification and synthesis of correct-by-construction 
software.
They require a top-down approach: programmers first write choreographies, and then use them to verify or 
synthesize their programs.
However, most software does not come with choreographies yet, which prevents their application.
To attack this problem, previous work investigated choreography extraction, which automatically constructs a 
choreography that describes the behavior of a given set of programs or protocol specifications.

We propose a new extraction methodology that improves on the state of the art:
we can deal with programs that are equipped with state and internal
computation; time complexity is dramatically better;
and we capture programs that work by exploiting asynchronous communication.


\end{abstract}

\section{Introduction}
\label{sec:intro}
Choreographies are global descriptions of interactions among components.
They have been used as a basis for different models and tools that aim at tackling the complexity of modern 
software, where separate units -- such as processes, objects, and threads -- interact to reach a common 
goal~\cite{BPMN,CDL}.

Two lines of research
are of particular interest.
In \emph{choreography specifications}, choreographies specify
interaction protocols, e.g.,
multiparty session types~\cite{HYC16}.
In \emph{choreographic programming}~\cite{M13:phd}, choreographies are programs that 
define the behavior 
of concurrent algorithms~\cite{CM16} and/or distributed systems~\cite{CHY12,CM13,PGGLM15}.
The key feature of these works is EndPoint Projection (EPP), a procedure that translates choreographies to 
correct endpoint behaviors in lower-level models. For choreography specifications, EPP generates the local 
specifications of each participant; these specifications 
can then be used for verification, to check whether some programs implement their role in the protocol correctly and 
will thus interact without problems at runtime~\cite{HYC16}. In 
choreographic programming, instead, EPP generates correct-by-construction implementations in a model of 
executable code (program synthesis), typically given in terms of a process calculus~\cite{CM13}.

EPP implements a top-down development methodology:
developers first write choreographies and then use 
the output mechanically generated by EPP. However, there are scenarios 
where this methodology is not applicable; for example:
\begin{itemize}
\item Analysis or integration of legacy software:
either code developed previously, or new code 
written in a technology without support for choreographies.
\item  Updates: endpoint programs generated by EPP can later be updated locally (e.g., for
configuration or optimizations). Since the original choreography is not automatically 
updated, rerunning EPP loses these changes.
\end{itemize}

To attack these issues, previous work investigated a procedure to infer choreographies from arbitrary 
endpoint descriptions. We call this procedure \emph{choreography extraction}. To 
the best of our knowledge, the current reference for extracting choreography specifications
is~\cite{LTY15}, where graphical choreographies that represent protocol specifications are extracted from 
communicating automata~\cite{BZ83}. Instead, the state of the art for extraction in choreographic 
programming is~\cite{CMS14}, where extraction takes terminating processes typed using a fragment of linear 
logic as input. We advance both lines of work in several aspects, described below. 

\subsection{Contributions}

\smallpar{Extraction for synchronous systems}
We define an extraction procedure that applies directly to both choreography 
specifications and choreographic programming, by working with representative models.
%
We focus on the more difficult case of choreographic 
programming, and then show how our approach can be applied 
to other settings in \S~\ref{sec:extensions}.
First we define an extraction algorithm for processes with synchronous communications
(\S~\ref{sec:extraction}), which 
showcases the key elements of our construction: building a choreography corresponds to finding
paths in a graph that represents the abstract execution of the input processes. Our extraction also helps in debugging: 
if extraction detects
a potential deadlock, we pinpoint it with a special term ($\dlock$).
This is the first extraction procedure for choreographic programming that can deal with procedures and infinite 
behavior~\cite{CMS14}.

\smallpar{Asynchrony}
We extend our development to asynchronous communication 
(\S~\ref{sec:async}).
The key novelty is that we can extract a new class of behaviors where processes progress because 
of asynchronous communication. The simplest example of this class is a two-way exchange: a network of two processes 
where each process 
starts by sending a value to the other, and then consumes the received value. This network is deadlocked under a 
synchronous semantics, violating the state-of-the-art requirements for extraction~\cite{LTY15}.
Capturing these behaviors is challenging for two reasons: there is no choreography language capable of 
representing them; and the extraction algorithms presented so far require the behaviors of processes to be 
representable also under a synchronous interpretation.
We overcome both limitations with a new choreography primitive for multiparty asynchronous exchange and 
a look-ahead mechanism for asynchronous actions in extraction.

\smallpar{Efficiency}
We show that our extraction has exponential worst-case time complexity in both the synchronous and the asynchronous 
cases (\S~\ref{sec:extraction} and \S~\ref{sec:async}, respectively),
unlike the factorial case of~\cite{LTY15},
even though we can capture a new class of behaviors.
In particular, we need only one phase of exponential complexity, while~\cite{LTY15}
uses multiple phases applied in sequence. The authors of~\cite{LTY15} detail only the complexities of their first two 
phases: the first 
has exponential complexity (but in a quantity larger than ours), while the second has factorial complexity in a 
function exponential in the size of the input.
Our better time complexity stems from the design of our 
process language, which does not 
allow non-deterministic receives from different channels,
and careful algorithm crafting.
Despite the restriction, we can still model interesting examples thanks to asynchronous exchange. In 
\S~\ref{sec:async}, we present a novel
formulation of the alternating 2-bit protocol, which is given in~\cite{DY12} and used in~\cite{LTY15}
as a motivating example. Our formulation is simpler and does not require threads
as in~\cite{LTY15}.

\section{Related Work}
\label{sec:related}

\smallpar{Choreographic Programming}
The state of the art for extraction in choreographic programming is~\cite{CMS14}, where synchronous processes with 
finite behavior are typed using the multiplicative-additive fragment of linear logic. Our approach is significantly more 
expressive, bringing support for recursion and asynchronous communication. Also, the proof theory
in~\cite{CMS14} requires that there are no cycles in the structure of connections among processes. 
We do not have this limitation.

\smallpar{Choreography Specifications}
To the best of our knowledge, the state of the art for extracting choreography specifications is~\cite{LTY15}, which 
captures more behaviors than previous works with similar objectives~\cite{LT12,MYH09}.

Extraction in~\cite{LTY15} is more restrictive wrt.\ to asynchrony, requiring all process traces and choices to be 
represented in the synchronous transition system of the network. Thus,
networks that are safe because of asynchronous communication are not extracted in~\cite{LTY15}. Instead, our 
extraction can deal with programs that use multiparty asynchronous exchange, where multiple processes 
exchange values by exploiting asynchronous communication.
As a consequence, we can extract the alternating 2-bit protocol implemented via asynchronous exchange in 
\S~\ref{sec:async}, which is deadlocked under a synchronous semantics and thus cannot be extracted in~\cite{LTY15}.
Our extraction is the first capturing systems that are not correctly approximated by 
synchronous semantics (cf.~\cite{BB11}). A precise characterization of the class of extractable systems is thus an 
interesting future direction.

To circumvent the limitation that asynchronous exchange is not supported, choreographies in~\cite{LTY15} support 
local concurrency: processes can have internal threads. This opens up for an alternative 
formulation of the alternating 2-bit protocol, where the two participants use two threads each.
However, these choreographies are harder to read. As an example, compare our 
choreography for the alternating 2-bit protocol in \S~\ref{sec:async} to that obtained with the automata in~\cite{LTY15} 
(given in~\cite{DY12}, Protocol~7 in Example~2.1). Our formulation is a simple recursive procedure with two exchanges, 
whereas the control flow in~\cite{DY12} is rather intricate and uses three different operators (fork, join, and merge) 
at different places to compose two separate loops. In our opinion, our choreographies follow the principles of 
structured programming to a greater extent, and are simpler; also because coordination happens only through 
communication.

More interestingly than readability, local concurrency makes the complexity of 
extraction blow up factorially~\cite{LTY15}: process threads are represented using 
non-determinism between different actions in communicating automata. Determining whether the 
non-deterministic behavior of these automata is extractable takes (super-)factorial time
(factorial time in the size of a graph similar to our AES, cf.\ Definition~\ref{defn:extr-graph})!
Thus, asynchronous exchange supports a more efficient way of capturing an interesting class of behaviors.
Nevertheless, we believe that developing efficient extractions of local concurrency may be useful future 
work.

\section{Core Choreographies and Stateful Processes}
\label{sec:cc}
We review the languages of Core Choreographies (CC) and Stateful Processes (SP), 
from~\cite{CM16b}, which respectively model choreographies and endpoint programs.
We introduce labels in the reduction semantics for these calculi
to formalize the link between choreographies and their process implementations as a
bisimilarity.

\smallpar{Core Choreographies (CC)}
The syntax of CC is given in Figure~\ref{fig:cc_syntax}. A choreography $C$ describes the behavior of a set of 
processes ($\pid p$, $\pid q$, \ldots)\ running concurrently. Each process has an internal memory cell storing
a local value (the value of the process).
%
\begin{figure}[t]
\begin{align*}
    C & ::= \nil \acspar \eta; C \acspar \gencond \acspar \genrec \acspar \gencall
    \\
    \eta & ::= \gencom \acspar \gensel
    \hspace{25mm}
    e ::= v \acspar \pcont \acspar \ldots
\end{align*}
\caption{Core Choreographies, Syntax.}
\label{fig:cc_syntax}
\end{figure}
Term $\nil$ is the terminated choreography (omitted in examples).
Term $\eta;C$ reads ``the system executes $\eta$ and proceeds as $C$''.
An interaction $\eta$ is either: a value communication $\gencom$, where
process $\pid p$ evaluates $e$ and sends the result to process $\pid q$, which stores it in its
memory cell, replacing its previous value; or a selection $\gensel$, where $\pid p$ selects $l$ among the
branches offered by $\pid q$.
We abstract from the concrete language of expressions $e$, which models internal computation and is orthogonal
to our development, assuming only that: expressions can contain values $v$ and the placeholder $\pcont$, which
refers to the value of the process evaluating them; and evaluating expressions always terminates
and returns a value. In a conditional $\gencond$, $\pid p$ checks if its value is equal to $\pid q$'s to
decide whether the system proceeds as $C_1$ or $C_2$. Term $\genrec$ defines a procedure $X$ with body
$C_2$, which can be called in $C_1$ and $C_2$ by using term $X$.

The semantics of CC is given in terms of labeled reductions $C,\sigma \lto\lambda C',\sigma'$; the main reduction rules
are given in Figure~\ref{fig:cc_semantics}.
Reductions are also closed under context (procedure definitions) and under a structural precongruence $\precongr$, 
allowing procedure calls to be unfolded and non-interfering actions to be executed in any order.
The most interesting rule for $\precongr$ is rule $\rname{C}{Eta-Eta}$, which swaps communications between
disjoint sets of processes (modeling concurrency).
The total function $\sigma$ maps each process name to the value it stores.
Labels $\lambda$ 
tell us which action has been performed, which helps stating our later results.
\begin{figure}[t]
{\footnotesize
\begin{eqnarray*}
&\infer[\rname{C}{Com}]
{
	\gencom;C,\sigma
	\lto{\com{\pid p.v}{\pid q}}
	C, \sigma[\pid q \mapsto v]
}
{
	e[\sigma(\pid p)/\pcont] \eval v
}
\qquad
\infer[\rname{C}{Sel}]
{
	\gensel;C, \sigma \lto\gensel C, \sigma
}
{}
\\[1ex]
&
\infer[\rname{C}{Then}]
{
	\gencond, \sigma  \lto{\condlbl{\pid p}{\pid q}{\m{then}}}    C_1, \sigma
}
{
	\sigma(\pid p) = \sigma(\pid q)
}
\qquad
\infer[\rname{C}{Eta-Eta}]{
	\eta;\eta' \precongr \eta';\eta
}{\pn(\eta) \cap \pn(\eta') = \emptyset}
\end{eqnarray*}
}
\caption{Core Choreographies, Semantics and Structural Precongruence (selected rules).}
\label{fig:cc_semantics}
\end{figure}
In rule $\rname{C}{Com}$, $v$ is the value obtained by evaluating ($\eval$) the expression $e$, with $\pcont$
replaced by the value of the sender $\pid p$, $\sigma(\pid p)$. In the reductum, $\sigma$ is updated such that the 
receiver $\pid q$ stores $v$.
Rule $\rname{C}{Sel}$ does not alter $\sigma$:
selections model invoking a method/operation available at the receiver.
Rules $\rname{C}{Then}$ and $\rname{C}{Else}$ (omitted) model conditionals in the 
standard way.
Function $\pn(C)$ returns all the process names that appear in $C$, and $C \equiv C'$ means $C \precongr C'$
and $C' \precongr C$.

\begin{example}\label{ex:auth_cc}
We define a simple choreography for client authentication. We write $\sel{\pid p}{\pid c,\pid s}{l}$ as a
shortcut for $\sel{\pid p}{\pid c}{l}; \sel{\pid p}{\pid s}{l}$.
\[ \rec{X\!}{\!
	\left( \com{\pid c.pwd}{\pid a};
	\cond{\eqcom{\pid a}{\pid s}}{
	\left(\sel{\pid a}{\pid c,\pid s}{ok};
	\com{\pid s.t}{\pid c}\right)
	}{
	\left(\sel{\pid a}{\pid c,\pid s}{ko}; X\right)
	}\right)
}{X}
\]
In this choreography, a client process $\pid c$ sends a password to an authentication process $\pid a$, which
checks if the password matches that contained in the server-side process $\pid s$. If the 
password is correct, $\pid a$ notifies $\pid c$ and $\pid s$, and $\pid s$ sends an authentication token $t$ 
to $\pid c$. Otherwise, $\pid a$ notifies $\pid c$ and $\pid s$ that authentication failed, and a new attempt
is made (by recursively invoking $X$).\qed
\end{example}

\smallpar{Stateful Processes}
The calulus SP models concurrent/distributed implementations. Thus, unlike in 
CC, actions are now distributed among processes.

\begin{figure}[t]
\begin{align*}
B ::={} & \asend{\pid q}{e};B \ \mid \ \arecv{\pid p};B \ \mid \ \asel{\pid q}{l};B \ \mid \abranch{\pid p}{\{ l_i : 
B_i\}_{i\in I}} \ \mid \ &
N ::={} & \actor{\pid p}{}{B} \ \mid\  \emptyN\ \mid\ N \parp N \\
\mid{} & \lefteqn{\nil
 \mid \cond{\eqcom{\pcont}{\pid q}}{B_1}{B_2} \ \mid \ \rec{X}{B_2}{B_1} \ \mid \ \call X}
\end{align*}
\caption{Stateful Processes, Syntax.}
\label{fig:sp_syntax}
\end{figure}
The syntax of SP is given in Figure~\ref{fig:sp_syntax}.
Networks $N$ are parallel compositions of processes $\actor{\pid p}{}{B}$, read ``process $\pid p$ has 
behavior $B$''.
An output term $\asend{\pid q}{e};B$ sends the result of evaluating $e$ to 
$\pid q$, and then proceeds as $B$. Outputs are meant to synchronize with input terms at the target 
process, i.e., $\arecv{\pid p};B$, which receives a value from $\pid p$ to be stored locally and then proceeds as $B$.
Term $\asel{\pid q}{l};B$ sends the selection of the branch labeled $l$ to $\pid q$. Branches are offered by the 
receiver with term $\abranch{\pid p}{\{ l_i : 
B_i\}_{i\in I}}$, which offers a choice among the labels $l_i$ to $\pid p$. When one of these labels is selected, the 
respective behavior $B_i$ is run.
Term $\cond{\eqcom{\pcont}{\pid q}}{B_1}{B_2}$ communicates with process $\pid q$ to check whether it stores the same value as 
the process running this behavior, in order to choose between the continuations $B_1$ and $B_2$. Terms 
$\rec{X}{B_2}{B_1}$ and $\call X$ are procedure definition and call, respectively.

\begin{figure}[t]
\begin{eqnarray*}
&\infer[\rname{S}{Com}]
{
	\actor{\pid p}{}{\asend{\pid q}{e};B_1}
	\ \parp\ 
	\actor{\pid q}{}{\arecv{\pid p};B_2},\ \sigma
	\ \lto{\com{\pid p.v}{\pid q}} \ 
	\actor{\pid p}{}{B_1}
	\ \parp\ 
	\actor{\pid q}{}{B_2},\ \sigma[\pid q \mapsto v]
}
{
	e[\sigma(\pid p)/\pcont] \eval v
}
\\[1ex]
&\infer[\rname{S}{Sel}]
{
	\actor{\pid p}{}{\asel{\pid q}{l_j};B}
	\ \parp\ 
	\actor{\pid q}{}{\abranch{\pid p}{\{ l_i : B_i\}_{i\in I}}},\ \sigma
	\ \lto{\gensel} \
	\actor{\pid p}{}{B}
	\ \parp\ 
	\actor{\pid q}{}{B_j},\ \sigma
}
{j \in I}
\\[1ex]
&\infer[\rname{S}{Then}]
{
	\actor{\pid p}{}{\cond{\eqcom{\pcont}{\pid q}}{B_1}{B_2}}
	\ \parp\ 
	\actor{\pid q}{}{\asend{\pid p}{e};B'},\ \sigma
	\ \lto{\condlbl{\pid p}{\pid q}{\m{then}}} \ 
	\actor{\pid p}{}{B_1}
	\ \parp \ 
	\actor{\pid q}{}{B'},\ \sigma
}
{
	e[\sigma(\pid q)/\pcont]\eval \sigma(\pid p)
}
\end{eqnarray*}
\caption{Stateful Processes, Semantics (selected rules).}
\label{fig:sp_semantics}
\end{figure}
The semantics of SP is given by labeled reductions $N,\sigma \lto\lambda N',\sigma'$, with labels
$\lambda$ as in CC.\footnote{Deviating from~\cite{CM16b}, we model process values using $\sigma$ as for CC, for 
simplicity.}
Figure~\ref{fig:sp_semantics} shows the key rules (see the appendix for the complete set).
Two processes can synchronize when they refer to each other.
In rule $\rname{S}{Com}$, an output at $\pid p$ directed at $\pid q$ synchronizes with the dual input
action at $\pid q$ -- intention to receive from $\pid p$; in the reductum, $\pid q$'s value is updated.
The reduction receives the same label as the equivalent communication term in CC.
The other rules shown are similar.
The omitted rules are standard, and close the semantics under parallel composition, structural precongruence,
and procedure definitions.

\begin{example}\label{ex:auth_sp}
The following network implements the choreography in Example~\ref{ex:auth_cc}.
\begin{align*}
& \actor{\pid c}{}{\rec{X}{
	\asend{\pid a}{pwd};
	\abranch{\pid a}{\left\{ok: \arecv{\pid s},\ ko: X\right\}}
\ }{X}}
\\
\parp \ & \actor{\pid a}{}{\rec{X}{
	\arecv{\pid c};
	\cond{\eqcom{\pcont}{\pid s}}{
		\left( \asel{\pid c}{ok}; \asel{\pid s}{ok} \right)
	}{
		\left( \asel{\pid c}{ko}; \asel{\pid s}{ko}; X \right) 
	}
\ }{X}}
\\
\parp\ & \actor{\pid s}{}{\rec{X}{
	\asend{\pid a}{\pcont};
	\abranch{\pid a}{\left\{ ok: \asend{\pid c}{t},\ ko: X \right\}}
\ }{X}}
\end{align*}
\end{example}

\smallpar{EndPoint Projection (EPP)}
As shown in~\cite{CM16b}, there exists a partial function $\epp\cdot{}:\mbox{CC}\to\mbox{SP}$, called EndPoint
Projection (EPP), that produces correct implementations of choreographies.
EPP produces a parallel composition of processes, one for each process name in the original choreography:
$\epp{C}{}  =
\prod_{\pid p \in \pn(C)} \actor{\pid p}{}{\epp{C}{\pid p}}$.
The rules for 
computing $\epp{C}{}$ 
project the local action performed by the process of interest.
For example,
$\epp{\gencom}{\pid p}=\asend{\pid q}{e}$ and $\epp\gencom{\pid q}=\arecv{\pid p}$.

The network presented in Example~\ref{ex:auth_sp} is exactly the EPP of the choreography in Example~\ref{ex:auth_cc}.
Observe that the projection of the conditional in the original choreography for the processes $\pid c$ and $\pid s$ is 
a branching that supports all the possible choices made by process $\pid a$ in its projected conditional. 
Producing these branching terms is possible only if, whenever there is a conditional at a process ($\pid a$ in our 
example), all other processes receive a label that tells them which branch such a process has chosen. (In case the 
behaviors of the other processes are the same in both cases, producing branching terms is not necessary.)
When this cannot be done for a choreography $C$, the EPP for $C$ is undefined, and we say that $C$ is 
unprojectable. Conversely, $C$ is \emph{projectable} if $\epp{C}{}$ is defined.

In the remainder, we relate choreographies to network implementations via a strong labeled
reduction bisimilarity $\bisim$.
Bisimilarity is defined as usual~\cite{SW01}: it is the union of all bisimulation relations
$\mathcal R$, which in our case relate choreographies to networks.
A relation $\mathcal R$ is one such bisimulation if whenever $C \mathcal R N$ we have that, for all $\sigma$:
i) $C,\sigma \lto\lambda C',\sigma'$ implies $N,\sigma \lto\lambda N',\sigma'$ for some $N'$ with
$C'\mathcal R N'$; ii) $N,\sigma \lto\lambda N',\sigma'$ implies $C,\sigma \lto\lambda C',\sigma'$ for some
$C'$ with $C'\mathcal R N'$.
%
\begin{theorem}[adapted from~\cite{CM16b}]
If $C$ is projectable, then $C \bisim \epp{C}{}$.
\end{theorem}

\section{Extraction from SP}
\label{sec:extraction}

\smallpar{The finite case}
%
We first investigate \emph{finite SP}, the fragment of SP without recursive definitions, which we use to
discuss the intuition behind our extraction.

\begin{definition}
  \label{defn:extraction}
  We define a rewriting relation $\rwto$ on the language of CC extended with terms $\extract{N}$,
  where $N$ is a network in finite SP, as the transitive closure of:
  \begin{eqnarray*}
    &
    \infer{\extract N\rwto\nil}{N\ \equiv\ \nil}
    \qquad
    \infer{\extract N\rwto\gencom;\extract{N_p\parp N_q\parp N'}}
          {N\ \equiv\ \actor{\pid p}{}{\bsend{\pid q}{e}};N_p\parp\actor{\pid q}{}{\brecv{\pid p}};N_q\parp N'}
    \\
    &
    \infer{\extract N\rwto\sel{\pid p}{\pid q}{l_k};\extract{N_p\parp N_{q_k}\parp N'}}
          {N\ \equiv\ \actor{\pid p}{}{\bsel{\pid q}{l_k}};
          N_p\parp\actor{\pid q}{}{\bbranch{\pid p}{l_1:N_{q_1},\ldots,l_n:N_{q_n}}}\parp N'}
    \\
    &
    \infer{\extract N\rwto\cond{\eqcom{\pid p}{\pid q}}{\extract{N_{p_1}\parp N_q\parp N'}}{\extract{N_{p_2}\parp N_q\parp N'}}}
          {N\ \equiv\ \actor{\pid p}{}{\cond{\eqcom{\pcont}{\pid q}}{N_{p_1}}{N_{p_2}}}
          \parp\actor{\pid q}{}{\bsend{\pid p}{e}};N_q\parp N'}
    \qquad
    \infer{\extract N\rwto\dlock}{\mbox{no other rule applies}}
  \end{eqnarray*}

  A network $N$ in finite SP \emph{extracts} to a choreography $C$ if $\extract N\rwto C$.
\end{definition}


The last rule guarantees that every network is extractable.
Extraction uses structural precongruence (namely, commutativity and associativity of parallel
composition) to find matching actions.
For finite SP, this is not a problem (the set of networks equivalent to a given one is finite),
but it
makes extraction nondeterministic, e.g.,
the network
  $\actor{\pid p}{}{\bsend{\pid q}{e}}
  \parp\actor{\pid q}{}{\brecv{\pid p}}
  \parp\actor{\pid r}{}{\bsend{\pid s}{e'}}
  \parp\actor{\pid s}{}{\brecv{\pid r}}$
extracts both to $\com{\pid p.e}{\pid q};\com{\pid r.e'}{\pid s}$ and $\com{\pid r.e'}{\pid s};\com{\pid p.e}{\pid q}$.
These choreographies are equivalent by Rule~\rname{C}{Eta-Eta} (Figure~\ref{fig:cc_semantics}).
This holds in general, as stated below.

\begin{lemma}
  \label{lem:sound-fin}
  If $\extract N \rwto C_1$ and $\extract N \rwto C_2$, then $C_1 \equiv C_2$.
\end{lemma}

There is one important design option to consider:
what to do
with actions that cannot be matched, i.e., processes that will deadlock.
There are two alternatives:
restrict extraction to lock-free networks (networks where all processes eventually progress, in the sense
of~\cite{CDM14});
or extract stuck processes to a new choreography term $\dlock$, with the same semantics as $\nil$.
We choose the latter option for debugging reasons.
Specifically, practical applications of extraction may annotate $\dlock$ with the code of the deadlocked
processes, giving the programmer a chance to see exactly where the system is unsafe, and attempt at fixing it
manually.
Better yet: since the code to unlock deadlocked processes in process calculi can be efficiently
synthesized~\cite{CDM14}, our method may be integrated with the technique in~\cite{CDM14} to suggest an
automatic system repair.

\begin{remark}
If $\extract N \rwto C$ and $C$ does not contain $\dlock$, then $N$ is lock-free. However, even if $C$ 
contains $\dlock$, $N$ may still be lock-free: the code causing the deadlock 
may be dead code in a conditional branch that is never chosen during execution.
\end{remark}

Extraction is sound: it yields a choreography that is bisimilar to the original network. Also, for finite 
SP, it behaves as an inverse of EPP.
\begin{theorem}
  \label{thm:correct}
  Let $N$ be in finite SP.
  If $\extract N\rwto C$, then $C\bisim N$.
  Furthermore, if $N = \epp{C'}{}$ for some $C'$, then $\extract N \rwto C'$.
\end{theorem}
As we show later, the second part of this theorem does not hold in the presence of recursive definitions.

%
We now restate extraction in terms of a particular graph, which is the hallmark
of our development: when we add recursion to SP, we can no longer define extraction as a set of rewriting
rules.
We first introduce a new abstract semantics for networks, $N \lto\alpha N'$, defined as in
Figure~\ref{fig:sp_semantics} except for the rules for value communication and conditionals, which are
replaced by those in Figure~\ref{fig:asp_semantics} (we omit the obvious rule \rname{S}{Else}).
In particular, conditionals are now nondeterministic.
Labels $\alpha$ are like $\lambda$ but may now contain expressions (see the new rule \rname{S}{Com}); in all
other rules, $\lambda$ is replaced by $\alpha$.
We write $N\lmto{\til\alpha} N'$ for $N \lto{\alpha_1} \cdots \lto{\alpha_n} N'$.
\begin{figure}[t]
\begin{eqnarray*}
&
\infer[\rname{S}{Com}]
{
	\actor{\pid p}{}{\asend{\pid q}{e};B_1}
	\ \parp\ 
	\actor{\pid q}{}{\arecv{\pid p};B_2}
	\ \lto{\com{\pid p.e}{\pid q}} \ 
	\actor{\pid p}{}{B_1}
	\ \parp\ 
	\actor{\pid q}{}{B_2}
}{}
\\[1ex]
&
\infer[\rname{S}{Then}]{
	\actor{\pid p}{}{\cond{\eqcom{\pcont}{\pid q}}{B_1}{B_2}}
	\ \parp\ 
	\actor{\pid q}{}{\asend{\pid p}{e};B'}
	\ \lto{\condlbl{\pid p}{\pid q}{\m{then}}} \ 
	\actor{\pid p}{}{B_1}
	\ \parp \ 
	\actor{\pid q}{}{B'}
}{}
\end{eqnarray*}
\caption{Stateful Processes, Abstract Semantics (selected rules).}
\label{fig:asp_semantics}
\end{figure}
\begin{definition}
  \label{defn:extr-graph}
  Given a network $N$, the \emph{Abstract Execution Space (AES)} of $N$ is the directed graph
  obtained by considering all possible abstract reduction paths from $N$.
  Its vertices are all the networks $N'$ such that $N\lmto{\til\alpha} N'$, and there is an edge between two 
  vertices $N_1$ and $N_2$
  labeled $\alpha$ if $N_1\lto\alpha N_2$.
  
  A \emph{Symbolic Execution Graph (SEG)} for $N$ is a subgraph of its AES that contains $N$ and such 
  that each vertex $N'\not\preceq\nil$ has either one outgoing edge labeled by an
  $\eta$ or two outgoing edges labeled $\eqcom{\pid p}{\pid q}:\m{then}$ and $\eqcom{\pid p}{\pid q}:\m{else}$.
\end{definition}

Intuitively, the AES of $N$ represents all possible evolutions of $N$ (each evolution is a path in this graph).
A SEG fixes the order of execution of actions, but still abstracts from the state (and
thus considers both branches of conditionals).
For networks in finite SP, these graphs are finite.

Given a network $N$, there is a one-to-one correspondence between SEGs for $N$ and choreographies $C$ such
that $\extract N\rwto C$.
Indeed, given a SEG we can extract a choreography as follows.
We start from the initial vertex, labeled $N$.
If there is an outgoing edge with label $\eta$ to $N'$, we add $\eta$ to the choreography and continue from
$N'$.
If there are two outgoing edges with labels $\eqcom{\pid p}{\pid q}:\m{then}$ and
$\eqcom{\pid p}{\pid q}:\m{else}$ to $N_1$ and $N_2$, respectively, we extract a conditional whose
branches are the choreographies extracted by continuing exploration from $N_1$ and $N_2$, respectively.
When we reach a leaf, we extract $\nil$ or $\dlock$, according to whether its label is equivalent to $\nil$ or
not.
Conversely, we can build a SEG from a particular rewriting of $\extract N$ by following the choreography
actions one at a time.

\smallpar{Treating recursive definitions}
%
We now extend extraction to networks with recursive definitions, using SEGs.
We need to be careful with the definition of the AES, since including all possible (abstract) executions now
may make it infinite (due to recursion unfolding), and thus extraction may not terminate.
To avoid this, we only allow recursive definitions to be unfolded (once)
if they occur at the head of a process involved in a reduction.
With this restriction, we can define the AES and SEGs for a network as in the finite case.
These graphs may now contain cycles: a network may evolve into the same term after
a few reductions.

\begin{example}
  \label{ex:aes}
  Consider the following network.
  \begin{align*}
    & \actor{\pid p}{}{{}\rec{X}{\asend{\pid q}{\pcont};\bbranch{\pid q}{\lleft:\asend{\pid 
q}{\pcont};X,\lright:\nil}}{\asend{\pid q}\pcont;X}} \\
  \parp &
  \actor{\pid q}{}{{}\rec{Y}{\arecv{\pid p};\arecv{\pid p};\cond{\eqcom{\pcont}{\pid r}}{\bsel{\pid 
p}{\lleft};Y}{\bsel{\pid p}{\lright};\nil}}{Y}} \parp
  \actor{\pid r}{}{{}\rec{Z}{\asend{\pid q}{\pcont};Z}{Z}}
  \end{align*}

  This network generates the AES in Figure~\ref{fig:aes}, which 
  is also its SEG.\qed
\end{example}
\begin{figure}[t]
  \[\xymatrix@R=1em{
    & \makebox[0cm][c]{$\actor{\pid p}{}{\asend{\pid q}{\pcont};X} \parp \actor{\pid q}{}{Y} \parp \actor{\pid r}{}{Z}$}
    \ar[d]^{\com{\pid p.\pcont}{\pid q}}
    \\
    & \makebox[0cm][c]{$\actor{\pid p}{}{X} \parp \actor{\pid q}{}{\arecv{\pid p};\cond{\eqcom{\pcont}{\pid 
r}}{\bsel{\pid p}{\lleft};Y}{\bsel{\pid p}{\lright};\nil}} \parp \actor{\pid r}{}{Z}$}
    \ar[d]^{\com{\pid p.\pcont}{\pid q}}
    \\
    & {\begin{array}c \displaystyle\actor{\pid p}{}{\bbranch{\pid q}{\lleft:\asend{\pid q}{\pcont};X,\lright:\nil}} \parp {} \\
        \displaystyle \actor{\pid q}{}{\cond{\eqcom{\pcont}{\pid r}}{\bsel{\pid p}{\lleft};Y}{\bsel{\pid p}{\lright};\nil}} \parp \actor{\pid r}{}{Z}
        \end{array}}
    \ar '[]+L `l[ddl]_{\eqcom{\pid q}{\pid r}.\m{then}} [ddl]
    \ar[d]^{\eqcom{\pid q}{\pid r}.\m{else}}
    \\
    & \makebox[8em][c]{$\actor{\pid p}{}{\bbranch{\pid q}{\lleft:\asend{\pid q}{\pcont};X,\lright:\nil}} \parp 
\actor{\pid q}{}{\bsel{\pid p}{\lright};\nil} \parp \actor{\pid r}{}{Z}$}
    \ar[d]^{\sel{\pid q}{\pid p}{\lright}}
    \\
    \makebox[6em][c]{$\actor{\pid p}{}{\bbranch{\pid q}{\lleft:\asend{\pid q}{\pcont};X,\lright:\nil}} \parp 
\actor{\pid q}{}{\bsel{\pid p}{\lleft};Y} \parp \actor{\pid r}{}{Z}$}
    \ar '[]+UL `u[uuuur]-<6em,0em>^{\sel{\pid q}{\pid p}{\lleft}} [uuuur]-<6em,0em>
    & \makebox[0em][c]{$\actor{\pid p}{}{\nil} \parp \actor{\pid q}{}{\nil} \parp \actor{\pid r}{}{Z}$}
  }\]
  \caption{The AES and SEG for the network in Example~\ref{ex:aes}.}
  \label{fig:aes}
\end{figure}
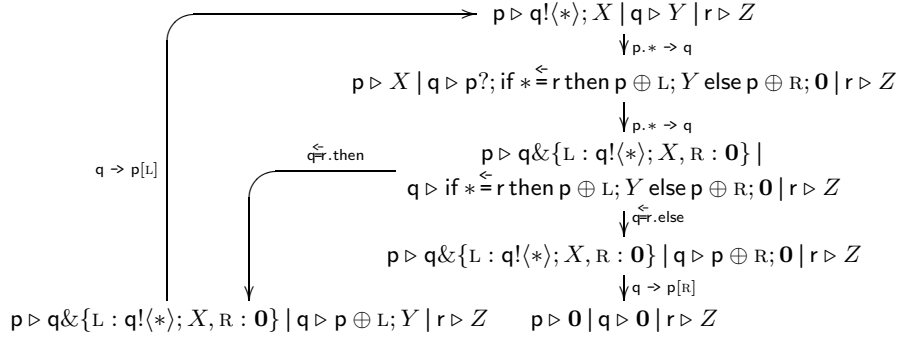

The key insight is that the definitions of recursive
procedures are extracted from the loops in the SEG, rather than from the recursive definitions in the source
network.
This construction typically yields mutually recursive definitions, motivating a small change to CC that does not add 
expressivity:
we replace the constructor $\genrec$ by top-level procedure definitions, in the style of~\cite{ourPCstuff}.
A choreography now becomes a pair $\langle\mathcal D,C\rangle$, where $\mathcal D=\{X_i=C_i\}$ and all
procedure calls in either $C$ or the $C_i$ are to some $X_i$ defined in $D$.

\begin{definition}\label{def:extr_unsound}
The choreography extracted from a SEG is defined as follows.
We annotate each node that has more than one incoming edge with a unique procedure identifier.
Then, for every node annotated with an identifier, say $X$, we replace each of its incoming edges with an edge to a new 
leaf node that contains a special term $X$ (so now the node annotated with $X$ has no incoming edges).
This eliminates all loops in the SEG, allowing us to reuse the extraction procedure for the non-recursive case to 
extract the desired pair $\langle\mathcal D,C\rangle$.
We get $C$ by extraction starting from the initial network.
Then, for each node that we annotated with an $X$, we extract a choreographic procedure $X$ in $\mathcal D$ that has as 
body the choreography extracted from the graph that starts from that annotated node. Any new leaf node containing a 
special term $X$ is extracted as a procedure call $X$.
\end{definition}

\begin{example}
  Consider the SEG in Figure~\ref{ex:aes}.
  To extract a choreography, we annotate the topmost node with a procedure identifier $X$ and replace the incoming
  edge to that node with an edge to a new leaf $X$.
  We thus extract $X$ to be
  \[
  \com{\pid p.\pcont}{\pid q};\com{\pid p.\pcont}{\pid q};\cond{\eqcom{\pid q}{\pid r}}{\sel{\pid q}{\pid 
p}\lleft;X}{\sel{\pid q}{\pid p}\lright;\dlock}
  \]
  and the extracted choreography itself is simply $X$.
  The body of $X$ is not projectable (the branches for $\pid r$ are not mergeable, cf.~\cite{CM16b}), but it
  faithfully describes the behavior of the original network.\qed
\end{example}

The procedure in Definition~\ref{def:extr_unsound} always terminates, but sometimes
it yields choreographies that starve some processes.
%
As an example, the network
\begin{align}
  \label{this_is_a_label}
  & \actor{\pid p}{}{{}\rec{X}{\asend{\pid q}{\pcont};X}{X}}
  \ \parp\ \actor{\pid q}{}{{}\rec{Y}{\arecv{\pid p};Y}{Y}} \\
  \nonumber
  \parp \ &\actor{\pid r}{}{{}\rec{Z}{\asend{\pid s}{\pcont};Z}{Z}}
  \ \parp\ \actor{\pid s}{}{{}\rec{W}{\arecv{\pid r};W}{W}}
\end{align}


  
  has two SEGs, which extract
  to the choreographies
  $\rec{X}{\com{\pid p.\pcont}{\pid q};X}{X}$ and 
  $\rec{X}{\com{\pid r.\pcont}{\pid s};X}{X}$, 
  none of which captures all the behaviors of $N$.

To avoid this problem, we change the definitions of AES and SEGs slightly.
We annotate all procedure calls in networks with either $\circ$ or $\bullet$.
The node in the AES corresponding to the initial network has all procedure calls annotated with $\circ$.
There is an edge from $N$ to $N'$ with label $\alpha$ if $N\lto\alpha N'$ and the procedure calls in $N'$ are 
annotated as follows.
\begin{itemize}
\item If executing $\alpha$ does not require unfolding procedure calls, then all calls in $N'$ are annotated
  as in $N$.
\item If executing $\alpha$ requires unfolding procedure calls, then we annotate all the calls in $N'$
  introduced by these unfoldings with $\bullet$.
  If $N'$ now has \emph{all} procedure calls annotated with $\bullet$, we change all annotations to $\circ$.
\end{itemize}

We then require loops in a SEG to contain a node where every procedure call is annotated with $\circ$.
This ensures that every procedure call is unfolded at least once before returning to the same node.
This holds even if $\actor{\pid p}{}{X}$ unfolds to a behavior that calls different procedures,
but not $X$: in order to return to the same node, the newly invoked procedures themselves need to be 
unfolded.

\begin{example}
  The annotated AES for the network~\eqref{this_is_a_label} is:
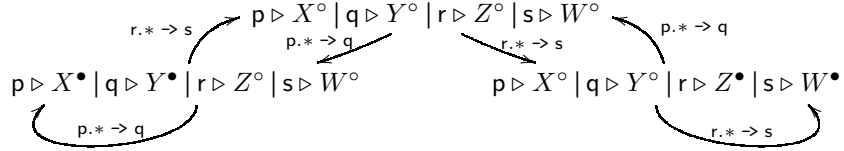
  \[\xymatrix@R=1em{
    &\makebox[2em][c]{$\actor{\pid p}{}{X^\circ}\parp\actor{\pid q}{}{Y^\circ}\parp\actor{\pid 
r}{}{Z^\circ}\parp\actor{\pid s}{}{W^\circ}$}
    \ar@/^/[dl]_(.4){\com{\pid p.\pcont}{\pid q}}
    \ar@/_/[dr]^(.4){\com{\pid r.\pcont}{\pid s}}
    \\
    \makebox[10em][c]{$\actor{\pid p}{}{X^\bullet}\parp\actor{\pid q}{}{Y^\bullet}\parp\actor{\pid 
r}{}{Z^\circ}\parp\actor{\pid s}{}{W^\circ}$}
    \ar@/^/[ur]+DL+<-5.7em,.7em>^{\com{\pid r.\pcont}{\pid s}}
    \ar@(dr,dl)[]+DL_{\com{\pid p.\pcont}{\pid q}}
    &&\makebox[10em][c]{$\actor{\pid p}{}{X^\circ}\parp\actor{\pid q}{}{Y^\circ}\parp\actor{\pid 
r}{}{Z^\bullet}\parp\actor{\pid s}{}{W^\bullet}$}
    \ar@/_/[ul]+DR+<5.7em,.7em>_{\com{\pid p.\pcont}{\pid q}}
    \ar@(dl,dr)[]+DR^{\com{\pid r.\pcont}{\pid s}}
  }\]
  This AES now has the following two SEGs:
  \[\xymatrix@R=1em{
    \actor{\pid p}{}{X^\circ}\parp\actor{\pid q}{}{Y^\circ}\parp\actor{\pid r}{}{Z^\circ}\parp\actor{\pid s}{}{W^\circ}
    \ar@/^/[d]^{\com{\pid p.\pcont}{\pid q}}
    &&
    \actor{\pid p}{}{X^\circ}\parp\actor{\pid q}{}{Y^\circ}\parp\actor{\pid r}{}{Z^\circ}\parp\actor{\pid s}{}{W^\circ}
    \ar@/^/[d]^{\com{\pid r.\pcont}{\pid s}}
    \\
    \actor{\pid p}{}{X^\bullet}\parp\actor{\pid q}{}{Y^\bullet}\parp\actor{\pid r}{}{Z^\circ}\parp\actor{\pid 
s}{}{W^\circ}
    \ar@/^/[u]^{\com{\pid r.\pcont}{\pid s}}
    &&
    \actor{\pid p}{}{X^\circ}\parp\actor{\pid q}{}{Y^\circ}\parp\actor{\pid r}{}{Z^\bullet}\parp\actor{\pid 
s}{}{W^\bullet}
    \ar@/^/[u]^{\com{\pid p.\pcont}{\pid q}}
  }\]
  Observe that the self-loops are discarded because they do not go through a node with all $\circ$ annotations.
  From these SEGs, we can extract two definitions for $X$:
  \[
  \rec{X}{\com{\pid p.\pcont}{\pid q};\com{\pid r.\pcont}{\pid s};X}{X}
  \qquad\mbox{ and }\qquad
  \rec{X}{\com{\pid r.\pcont}{\pid s};\com{\pid p.\pcont}{\pid q};X}{X}
  \]
  
  Both of these definitions correctly capture all behaviors of the network.\qed
\end{example}

A similar situation may occur if there are processes with finite behavior (no procedure calls): the network
\[
\actor{\pid p}{}{\rec{X}{\asend{\pid q}{\pcont};X}{X}}
\ \parp\ 
\actor{\pid q}{}{\rec{Y}{\arecv{\pid p};Y}{Y}}
\ \parp\ 
\actor{\pid r}{}{\asend{\pid s}{\pcont}}
\ \parp\ 
\actor{\pid s}{}{\arecv{\pid r}}
\]
can be extracted to the choreography $X$, with $X=\com{\pid p.\pcont}{\pid q};X$, where $\pid r$ and $\pid s$
never communicate.
Hence, we require that if a node in a SEG has more than one incoming edge (it is a ``loop'' node)
and contains processes with finite behavior, then these processes must be deadlocked 
(being finite, this is trivially verifiable).
This ensures that if finite processes are able to reduce, they cannot be in a loop.

\begin{definition}
  \label{defn:valid-seg}
  A SEG for a network $N$ is valid if all its loops:
  \begin{itemize}
  \item pass through a node where all recursive calls are marked with $\circ$;
  \item start in a node where all processes with finite behavior are deadlocked.
  \end{itemize}
  A network $N$ extracts to a choreography $C$ if $C$ can be constructed (as in
  Definition~\ref{def:extr_unsound}) from a valid SEG for $N$.
\end{definition}

Validity implies, however, that there are some non-deadlocked networks that are not
extractable, such as
\[
\actor{\pid p}{}{\rec{X}{\asend{\pid q}{\pcont};X}{X}}
\parp
\actor{\pid q}{}{\rec{Y}{\arecv{\pid p};Y}{Y}}
\parp
\actor{\pid r}{}{\rec{Z}{\arecv{\pid p};Z}{Z}}
\]
for which there is no valid SEG.
This is to be expected, since deadlock-freedom is undecidable in SP.
We can generalize this observation as a necessary condition for extraction to be defined, in the following theorem.
\begin{theorem}
  \label{thm:aes-good}
  If the AES for a network $N$ does not contain nodes from which a process is always deadlocked, then $N$ is
  extractable.
\end{theorem}

Lemma~\ref{lem:sound-fin} and the first part of Theorem~\ref{thm:correct} still hold for extraction in SP with
recursion, but the second part of Theorem~\ref{thm:correct} does not: in general, the projection of a
choreography is extracted to a choreography with different procedures, since extraction ignores the actual
definitions in the source network.
%
\begin{theorem}
  \label{thm:oc-ac-ap}
  If $C$ is a choreography extracted from a network $N$, then $N\bisim C$.
\end{theorem}

We conclude this section with some complexity theoretical considerations.

\begin{lemma}
\label{lem:size-aes}
The annotated AES for a network of size $n$ has at most $e^{\frac{2n}{e}}$ vertices.
\end{lemma}

\begin{theorem}
  \label{thm:extract-complex}
  Extraction from a network of size $n$ terminates in time $O(n e^{\frac{2n}{e}})$.
\end{theorem}

As discussed earlier, this time complexity is a dramatic improvement over
earlier, comparable work. However, in practice, we may be able to
perform even better.
Algorithmically, all the required work stems from traversals of the AES,
so any reduction in its (explored) size will lead to proportional
runtime improvements.
Thus, instead of first computing the entire AES and then a
valid SEG, we can compute the relevant parts of the AES lazily as we need them,
so parts of the AES that are never explored while computing a valid SEG are
never generated.


\section{Asynchrony}
\label{sec:async}
We now discuss an asynchronous semantics for SP, with which we can express new safe behaviors.
Most notably, SP can now express asynchronous exchange (Example~\ref{ex:async-mot}).
We also show a novel choreography primitive that successfully captures this pattern, which cannot be described
in previous works on choreographic programming, and extend our algorithm to extract it from networks.

\smallpar{Asynchronous SP}
Asynchronous communication can be added to SP using standard techniques for process calculi.
In the semantics of networks, we add a FIFO queue for each pair of processes.
Communications now synchronize with these queues: send actions append a message in the queue 
of the receiver, and receive actions remove the first message from the queue of the receiver
(see~\cite{ourPCstuff} for a formalization in an extension of SP).

\begin{example}
  \label{ex:async-mot}
  The network
  $
  \actor{\pid p}{}{\asend{\pid q}{\pcont};\arecv{\pid q}}
  \ \parp\ 
  \actor{\pid q}{}{\asend{\pid p}{\pcont};\arecv{\pid p}}
  $ exemplifies the pattern of asynchronous exchange.
  This network is deadlocked in synchronous SP, but runs without errors in asynchronous SP: both $\pid p$ and
  $\pid q$ can send their respective values, becoming ready to receive each other's messages.
  This behavior is not representable in any previous work on choreographies (including CC from
  \S~\ref{sec:cc}), since all choreographies presented so far can only describe processes that are not
  deadlocked under a synchronous semantics (see~\cite{ourPCstuff} for a formal argument).\qed
\end{example}

\smallpar{The multicom}
The situation in Example~\ref{ex:async-mot} is prototypical of programs that are safe only in an asynchronous 
setting: a group of processes wants to send messages to a group of receivers, with circular
dependencies among communications.

We deal with this situation by means of a new choreography action, which we call a \emph{multicom}.
Syntactically, a multicom is a list of communication actions with distinct receivers, which we write 
$\genmulticometa$. 
In the unary case, we obtain the usual communications and selections; by removing these from the syntax of CC
and adding the multicom, we obtain a more expressive calculus with fewer primitives.
The semantics of multicom is given by the following rule, which generalizes (and replaces) both \rname{C}{Com}
and \rname{C}{Sel}.
\begin{displaymath}
\infer[\rname{C}{MCom}]
{
	\genmulticometa;C,\sigma
	\lto{
		\genmulticometa[e_i/v_i]_{i \in I}
	}
	C, \sigma[\pid q_i \mapsto v_i]_{i \in I}
}
{
	I = \left\{ i \mid \com{\pid p_i.e_i}{\pid q_i} \in \til \eta \right\}
	&
	v_i = e_i[\sigma(\pid p_i)/\pcont]
}
\end{displaymath}

Structural precongruence rules for the multicom are motivated by its intuitive semantics:
actions inside a multicom can be permuted as long as the senders differ, and sequential multicoms can be
merged as long as they do not share receivers and there are no sequential constraints between them (i.e., none
of the receivers in the first multicom is a sender in the second one).
\begin{eqnarray*}
  &\infer[\rname{C}{MCom-Perm}]
  {\multicom{\ldots, \eta_1, \eta_2, \ldots} \equiv \multicom{\ldots, \eta_2, \eta_1, \ldots}}
  {\pn(\eta_1)\cap\pn(\eta_2)=\emptyset}
  \\[1ex]
  &\infer[\rname{C}{MCom-MCom}]
  {\multicom{\til\eta};\multicom{\til\nu} \equiv \multicom{\til\eta,\til\nu}}
  {\rcv(\eta)\cap\rcv(\nu)=\emptyset
    &
    \rcv(\til\eta)\cap\snd(\til\nu)=\emptyset
  }
\end{eqnarray*}

From these rules we can derive all instances of \rname{C}{Eta-Eta}, e.g.:
\[
\com{\pid p.\pcont}{\pid q};\com{\pid r.\pcont}{\pid s}
\equiv
\multicom{\com{\pid p.\pcont}{\pid q}\\\com{\pid r.\pcont}{\pid s}}
\equiv
\multicom{\com{\pid r.\pcont}{\pid s}\\\com{\pid p.\pcont}{\pid q}}
\equiv
\com{\pid r.\pcont}{\pid s};\com{\pid p.\pcont}{\pid q}
\]


The problematic program in Example~\ref{ex:async-mot} can now be written as
$\multicom{\com{\pid p.\pcont}{\pid q}\\ \com{\pid q.\pcont}{\pid p}}$.

Structural precongruence rules for multicom also allow us to define a normal form for choreographies, where no
multicom can be split in smaller multicoms.

\smallpar{Extraction}
In order to extract choreographies containing multicoms, we alter the definition of the AES for a process
network by allowing multicoms as labels for the edges.
These can be computed using the following iterative algorithm.
\begin{enumerate}
\item For a process $\pid p$ with behavior $\asend{\pid q}e;B$ (or $\asel{\pid q}l;B$), set
  $\m{actions}=\emptyset$ and $\m{waiting}=\{\com{\pid p.e}{\pid q}\}$ (resp.\ 
  $\m{waiting}=\{\sel{\pid p}{\pid q}l\}$).
\item While $\m{waiting}\neq\emptyset$:
  \begin{enumerate}
  \item Move an action $\eta$ from $\m{waiting}$ to $\m{actions}$.
    Assume $\eta$ is of the form $\com{\pid r.e}{\pid s}$ (the case for label selection is similar).
  \item If the behavior of $\pid s$ is of the form $a_1;\ldots;a_k;\arecv{\pid r};B$ where each $a_i$ is
    either the sending of a value or a label selection, then: for each $a_i$, if the corresponding
    choreography action is not in $\m{actions}$, add it to $\m{waiting}$.
  \end{enumerate}
\item Return $\m{actions}$.
\end{enumerate}
This algorithm may fail (the behavior of $\pid s$ in step~2(b) is not of the required form), in which case the action
initially chosen cannot be unblocked by a multicom.

\begin{example}
  Consider the network from Example~\ref{ex:async-mot}.
  Starting with action $\asend{\pid q}\pcont$ at process $\pid p$, we initialize
  $\m{actions}=\emptyset$ and $\m{waiting}=\{\com{\pid p.\pcont}{\pid q}\}$.
  We pick the action $\com{\pid p.\pcont}{\pid q}$ from $\m{waiting}$ and move it to $\m{actions}$.
  The behavior of $\pid q$ is $\asend{\pid p}\pcont;\arecv{\pid p}$, which is of the form described in
  step~2(b); the choreography action corresponding to $\asend{\pid p}\pcont$ is $\com{\pid q.\pcont}{\pid p}$,
  so we add this action to $\m{waiting}$, obtaining $\m{actions}=\{\com{\pid p.\pcont}{\pid q}\}$ and
  $\m{waiting}=\{\com{\pid q.\pcont}{\pid p}\}$.
  Now we consider the action $\com{\pid q.\pcont}{\pid p}$, which we move from $\m{waiting}$ to $\m{action}$,
  and look at $\pid p$'s behavior, which is $\asend{\pid q}\pcont;\arecv{\pid q}$.
  The choreography action corresponding to $\asend{\pid q}\pcont$ is $\com{\pid p.\pcont}{\pid q}$, which is
  already in $\m{actions}$, so we do not change $\m{waiting}$.
  The set $\m{waiting}$ is now empty, and the algorithm terminates,
  returning $\multicom{\com{\pid p.\pcont}{\pid q}\\ \com{\pid q.\pcont}{\pid p}}$.
  We would obtain the equivalent
  $\multicom{\com{\pid q.\pcont}{\pid p}\\ \com{\pid p.\pcont}{\pid q}}$ by starting with the send action at
  $\pid q$.\qed
\end{example}

\begin{example}
As a more sophisticated example, we show how our new choreographies with multicom can model the alternating 2-bit 
protocol.
Here, Alice alternates between sending a $0$ and a $1$ to Bob; in turn, Bob sends an acknowledgment for 
every bit he receives, and Alice waits for the acknowledgment before sending another copy of the same bit.
Since we are in an asynchronous semantics, we only consider the time when the messages arrive.
With this in mind, we can write this protocol as the following network.
\begin{align*}
  \actor{\pid a}{}{{}&\rec{X}{(\arecv{\pid b};\asend{\pid b}0;\arecv{\pid b};\asend{\pid b}1;X)}{(\asend{\pid 
b}0;\asend{\pid b}1;X)}} \\
  \parp\ 
  \actor{\pid b}{}{{}&\rec{Y}{(\arecv{\pid a};\asend{\pid a}{\m{ack}_0};\arecv{\pid a};\asend{\pid a}{\m{ack}_1};Y)}{Y}}
\end{align*}

This implementation imposes exactly the dependencies dictated by the protocol.
For example, Alice can receive Bob's acknowledgment to the first $0$ before or after Bob receives the first $1$.
This network extracts to the choreography
\[
{\com{\pid a.0}{\pid b};X}
\qquad\mbox{where}\qquad
X = \multicom{\com{\pid a.1}{\pid b}\\ \com{\pid b.\m{ack}_0}{\pid a}};
  \multicom{\com{\pid a.0}{\pid b}\\ \com{\pid b.\m{ack}_1}{\pid a}}; X
\]
which is a simple and elegant representation of the alternating 2-bit protocol.\qed
\end{example}

Extraction for asynchronous SP is still sound, but behavioral equivalence is now an
expansion~\cite{Hennessy:whatever,SW01}, as each communication now takes two steps in asynchronous SP.
Its complexity is also no larger than for the synchronous case.
%
The algorithm computing the multicom takes
linear time in the size of the multicom produced.
Via a one-time preprocessing of the network, we can assume direct references from
communication terms in one process to the process it directs
its communication at, and from there to the current state of that process.
Other than the above, all constant steps in the algorithm can be
seen as an extension of the multicom.
Since adding a communication to a multicom removes
a potential node in the AES (as we are combining communications), the worst-case
time complexity is no worse than in the synchronous case.
In practice,
this complexity actually gets better when larger multicoms are created,
since building these is a much cheaper local operation
than exploring graphs that would be larger in terms of nodes
as well as edges without the multicoms.
%


\section{Extensions and Applications}
\label{sec:extensions}
We discuss some straightforward modifications of our extraction to cover other scenarios occurring in the literature.

\paragraph{More expressive communications and processes.}
In real-world contexts, the values stored and communicated by processes are typed,
and the receiver process can also specify how to treat incoming messages~\cite{ourPCstuff}.
This means that communication actions now have the form $\gencomf$, where $f$ is the function consuming the received 
message, and systems may deadlock because of typing errors.
Our construction applies without changes to this scenario.

Some works allow processes to store several values, used via variables~\cite{CHY12,CM13}.
Again, dealing with this situation does not require any changes to our algorithm.

\paragraph{Local conditionals.}
Many choreography models allow for a local conditional construct, i.e.,
$\cond{\pid p.e}{C_1}{C_2}$~\cite{CM13,MY13,PGGLM15}.
Dealing with this construct is simple: the \m{if} and \m{then} transitions 
now can occur whenever a process has a conditional as top action, since they no longer require synchronization with 
other processes.

\paragraph{Choreography Specifications.}
So far, we have considered choreographies that describe concrete implementations, i.e., processes are equipped with 
storage and local computational capabilities.
However, choreographies have also been advocated for the specification of communication protocols. Most notably, 
multiparty session types use choreographies to define types used in the verification of process calculi~\cite{HYC16}.
While there are multiple variants of multiparty session types, the one so far most used in practice is almost 
identical to a simplification of SP. In this variant, each pair of participants has a dedicated channel, and 
communication actions refer directly to the intended sender/recipient as in SP (see, e.g., the theory 
of~\cite{CM13,MY13,CLMSW16,CDYP16} and the practical implementations in~\cite{HMBCY11,NY14,M13:phd}).
To obtain multiparty session types from SP (and CC), we just need to: remove the capability of storing values at 
processes; replace message values with constants (representing types, which could also be extended to subtyping in 
the straightforward way); and make conditionals nondeterministic (since in types we abstract from the precise 
values and expression used by the evaluator).
These modifications do not require any significant change to our approach, since our AES already abstracts from 
data and thus our treatment of the conditional is already nondeterministic. For reference, we can simply treat 
the standard construct for an internal choice at a process $\pid p$ -- $C_1\oplus_{\pid p} C_2$ -- as syntactic 
sugar for a local conditional like $\cond{\pid p.\m{coinflip}}{C_1}{C_2}$.

%
%




\clearpage

\bibliography{biblio}


\end{document}